\documentclass[11pt,a4paper]{article}

\usepackage[utf8]{inputenc}
\usepackage[T1]{fontenc}
\usepackage{lmodern}                    
\usepackage{microtype}                  
\usepackage{setspace}
\usepackage[left=1.25in,right=1.25in,top=1.2in,bottom=1.2in]{geometry}
\usepackage{amsmath,amssymb,amsthm}
\usepackage{graphicx}
\usepackage[authoryear,round]{natbib}
\usepackage[dvipsnames]{xcolor}         
\usepackage{tikz}
\usetikzlibrary{calc}
\usepackage{enumitem}                   
\usepackage{parskip}
\usepackage{booktabs}                   
\usepackage{array}                      
\usepackage[font=small,labelfont=bf,skip=0.5em]{caption}
\usepackage{titlesec}                   

\titleformat{\section}
    {\Large\bfseries\sffamily\color{black}}
    {\thesection}{1em}{}
\titleformat{\subsection}
    {\large\bfseries\sffamily\color{black}}
    {\thesubsection}{1em}{}
\titleformat{\subsubsection}
    {\normalsize\bfseries\sffamily\color{black}}
    {\thesubsubsection}{1em}{}

\titlespacing*{\section}{0pt}{1.5em}{0.8em}
\titlespacing*{\subsection}{0pt}{1.2em}{0.6em}
\titlespacing*{\subsubsection}{0pt}{1em}{0.5em}

\setlist[itemize]{leftmargin=*,itemsep=0.3em,parsep=0pt,topsep=0.5em}
\setlist[enumerate]{leftmargin=*,itemsep=0.3em,parsep=0pt,topsep=0.5em}
\setlist[description]{leftmargin=*,itemsep=0.3em,parsep=0pt,topsep=0.5em}

\usepackage[most]{tcolorbox}

\definecolor{thmblue}{RGB}{240,248,255}      
\definecolor{thmframe}{RGB}{70,130,180}      
\definecolor{defgreen}{RGB}{240,255,240}     
\definecolor{defframe}{RGB}{60,140,100}      
\definecolor{assumamber}{RGB}{255,250,240}   
\definecolor{assumframe}{RGB}{210,140,50}    
\definecolor{neutralgray}{RGB}{248,248,248}  

\newtheoremstyle{thmstyle}
  {3pt}{3pt}{\itshape}{}{\bfseries\sffamily}{.}{.5em}{}
\newtheoremstyle{defstyle}
  {3pt}{3pt}{\normalfont}{}{\bfseries\sffamily}{.}{.5em}{}
\newtheoremstyle{remstyle}
  {3pt}{3pt}{\normalfont}{}{\bfseries}{.}{.5em}{}

\theoremstyle{thmstyle}
\newtheorem{theorem}{Theorem}[section]
\tcolorboxenvironment{theorem}{
    enhanced,
    breakable,
    colback=thmblue,
    colframe=thmframe,
    boxrule=0.8pt,
    sharp corners,
    left=10pt, right=10pt, top=8pt, bottom=8pt
}

\newtheorem{proposition}[theorem]{Proposition}
\tcolorboxenvironment{proposition}{
    enhanced,
    breakable,
    colback=thmblue,
    colframe=thmframe,
    boxrule=0.8pt,
    sharp corners,
    left=10pt, right=10pt, top=8pt, bottom=8pt
}

\theoremstyle{defstyle}
\newtheorem{definition}{Definition}[section]
\tcolorboxenvironment{definition}{
    enhanced,
    breakable,
    colback=defgreen,
    colframe=defframe,
    boxrule=0.8pt,
    sharp corners,
    left=10pt, right=10pt, top=8pt, bottom=8pt
}

\newtheorem{assumption}{Assumption}[section]
\tcolorboxenvironment{assumption}{
    enhanced,
    breakable,
    colback=assumamber,
    colframe=assumframe,
    boxrule=0.8pt,
    sharp corners,
    left=10pt, right=10pt, top=8pt, bottom=8pt
}

\theoremstyle{thmstyle}
\newtheorem{lemma}{Lemma}[section]
\tcolorboxenvironment{lemma}{
    enhanced,
    breakable,
    colback=neutralgray,
    colframe=black!40,
    boxrule=0pt,
    leftrule=3pt,
    sharp corners,
    left=12pt, right=10pt, top=5pt, bottom=5pt
}

\newtheorem{corollary}{Corollary}[section]
\tcolorboxenvironment{corollary}{
    enhanced,
    breakable,
    colback=neutralgray,
    colframe=black!40,
    boxrule=0pt,
    leftrule=3pt,
    sharp corners,
    left=12pt, right=10pt, top=5pt, bottom=5pt
}

\theoremstyle{defstyle}

\theoremstyle{remstyle}

\usepackage[
    colorlinks=true,
    linkcolor=RoyalBlue,
    citecolor=RoyalBlue,
    urlcolor=RoyalBlue,
    pdfborder={0 0 0},
    bookmarks=true,
    bookmarksopen=true,
    pdfstartview=FitH
]{hyperref}

\title{%
    \vspace{-1em}%
    {\LARGE\bfseries The Directions of Technical Change}%
    \thanks{We thank Gergely Szertics and seminar audience at CEU for helpful comments. This project was funded by the European Research Council (ERC Advanced Grant agreement number 101097789). The views expressed in this research are those of the authors and do not necessarily reflect the official view of the European Union or the European Research Council. Project no.\ 144193 has been implemented with the support provided by the Ministry of Culture and Innovation of Hungary from the National Research, Development and Innovation Fund, financed under the KKP\_22 funding scheme.}%
    \vspace{-0.5em}%
}

\author{%
    \begin{tabular}{c@{\hspace{2em}}c@{\hspace{2em}}c}
    \textbf{Miklós Koren} & \textbf{Zsófia L. Bárány} & \textbf{Ulrich Wohak}\\
    \small\textit{CEU, KRTK, CEPR, CESifo} & 
    \small\textit{CEU, CEPR, HSOE, ROA} & 
    \small\textit{CEU}\\
    \small miklos@koren.work & \small baranyzs@ceu.edu & \small wohaku@ceu.edu
    \end{tabular}
}

\date{\today}

\renewenvironment{abstract}{
    \begin{center}
    \begin{minipage}{0.9\textwidth}
    \small
    \setlength{\parskip}{0.5em}
}{
    \end{minipage}
    \end{center}
    \vspace{1em}
}

\begin{document}

\maketitle

\begin{abstract}
Generative AI is directional: it performs well in some task directions and poorly in others. Knowledge work is directional and endogenous as well: workers can satisfy the same job requirements with different mixes of tasks. We develop a high-dimensional model of AI adoption in which a worker uses a tool when it raises their output. Both the worker and the AI tool can perform a variety of tasks, which we model as convex production possibility sets. Because the tool requires supervision from the worker's own time and attention budget, adoption is a team-production decision, similar to hiring a coworker. The key sufficient statistics are the worker's pre-AI shadow prices: these equal the output gain from a small relaxation in each task direction, and they generally differ from the worker's observed activity mix. As AI capability improves, the set of adopted directions expands in a cone centered on these autarky prices. Near the entry threshold, small capability improvements generate large extensive-margin expansions in adoption. The model also delivers a structured intensive margin: between the entry and all-in thresholds, optimal use is partial. We parametrize the model in a simple but flexible way that nests most existing task-based models of technical change. 
\end{abstract}

\section{Introduction}
\label{sec:introduction}

Generative artificial intelligence (GenAI) based on large language models \citep{Vaswani2017} is not a single ``machine.'' It is a directional technology: it excels at some task combinations and performs poorly at others. Like other general-purpose technologies \citep{Eloundou2023}, it may reshape broad swaths of work, but two features make adoption and gains unusually hard to anticipate ex ante. First, adoption in knowledge work is often worker-driven: the interface is general-purpose and the main inputs are time and attention rather than specialized capital, so usage can spread inside organizations without explicit firm-level decisions.\footnote{On worker-driven adoption ("shadow IT"), see \citet{Klotz2019,Rakovic2020}; for rapid early diffusion, see \citet{UBS2023}.} Second, capabilities improve quickly but unevenly across tasks.\footnote{For software engineering benchmarks, see SWE-bench \citep{Jimenez2024}. The unevenness is described as a "jagged technological frontier" by \citet{DellAcqua2023} and quantified more broadly by \citet{Hendrycks2025}.} Together, worker-driven adoption and jagged capability mean that asking ``what is the effect of AI on labor?'' is ill-posed without specifying which tools help with which tasks---and for which workers.

Adoption is especially hard to characterize in knowledge work because workers can also change how they do the job. Job requirements may be stringent in terms of expected output, but different task mixes can deliver the same valued output. A software engineer can allocate time between writing new code, debugging, documentation, and code review; a management consultant between data analysis, client meetings, and slide preparation; a lawyer between legal research, drafting, and client communication. Reallocating effort across tasks is costly, however, and using AI consumes the same scarce time and attention that workers devote to their own work.

To understand AI adoption incentives in knowledge work, we develop a model in which what a worker can do in a unit of time is a convex set of feasible task combinations---a production possibility set (PPS).\footnote{This move is standard in international trade, where convex feasible sets and their supporting prices summarize production opportunities \citep{DixitNorman1980}.} A worker is employed in a job, which determines the value of any bundle of tasks \citep{Autor2013,Autor2025,LisePostelVinay2020}.\footnote{For evidence that workers have multidimensional skill endowments that affect occupational mobility, see \citet{GathmannSchoenberg2010}.} They choose a task mix to maximize output subject to a time budget. Skills determine which task combinations are feasible; job requirements determine which combinations generate high output.

We then treat AI as another production possibility set: a potential coworker with its own feasibility profile, and we model it as a directional technology that produces tasks in a fixed proportion. AI capability determines how much it can produce per unit of supervision time, in that fixed proportion. Using AI requires supervision, which uses time the worker could otherwise spend on their own tasks. The combined feasible set is therefore a weighted sum of what the worker and the tool can each produce, where the weights reflect how the worker allocates time between doing and supervising. Adopting AI is a collaboration decision rather than a simple substitution decision: the worker adopts when the expanded set is valuable enough to justify supervision.\footnote{The same logic extends to tools that can flexibly combine a small number of task directions in a high-dimensional space; see \citet{KaushikChaudhari2025}.}

The worker adopts an AI tool if using it raises their output. At their pre-AI optimum, each task direction has a shadow (autarky) price: the output gain from being able to do slightly more of that task, holding everything else fixed. These shadow prices summarize what is scarce in the worker--job match and map the AI tool's direction---a vector in task space---into a single scalar value. The worker adopts the AI tool when its direction is valuable at autarky prices. Autarky prices need not match the worker's activity vector: they reflect what they would like to shift toward at the margin, not what they currently do.

As AI capability improves, the set of adopted directions expands in a cone centered on the worker's autarky prices; near the entry threshold, small capability gains translate into large extensive-margin expansions in adoption. Adoption also has a structured intensive margin: a tool can be worth using but not worth using all the time, generating a region of stable partial adoption between an entry threshold and an all-in threshold. Because autarky prices depend jointly on worker skills and job requirements, adoption is match-specific: neither ``which workers adopt'' nor ``which jobs are affected'' has a well-defined answer without specifying the pairing. Finally, greater flexibility---in how the worker can allocate their time or in how different tasks are valued in output---widens the set of tool directions that are worth adopting, so adoption expands faster as capabilities improve.

We keep the analysis intentionally neoclassical to make the mechanism transparent. We study a worker doing a given job and abstract from governance and principal-agent frictions, so the worker maximizes output. The main objects and results---autarky prices, the adoption condition, the adoption cone, and the intensive margin---do not depend on a particular parametric form; they follow from convexity and marginal valuation.\footnote{Formally, we characterize task choice using a trade-style duality argument that delivers supporting prices for a convex feasible set \citep{DixitNorman1980}.}

For measurement and interpretation, we also provide a parametric illustration using standard constant-elasticity functional forms. We assume a constant-elasticity-of-substitution structure for how tasks combine in output, and a constant-elasticity-of-transformation (CET) structure \citep{PowellGruen1968} for how a worker's fixed time/attention budget can be used for different task bundles. This specialization yields closed-form expressions for autarky prices as functions of observable skill and requirement vectors and clarifies when skills versus job requirements dominate adoption predictions. In the limiting cases, rigid specialists have shadow prices pinned down by job requirements, while perfectly flexible workers have shadow prices pinned down by skills. In between, which we believe is the relevant case for knowledge work, both matter.

We hold the worker--job match fixed and analyze adoption of a given tool as an individual optimization problem. This isolates a building block that can be embedded in richer settings. In equilibrium, the availability of AI can reshape sorting and the distribution of jobs; multiple tools raise a team-composition problem; firms may govern usage and provide complementary investments \citep{Aldasoro2026}; and tool directions and capabilities may themselves respond to training data and market incentives.

The remainder of the paper proceeds as follows. Section \ref{sec:model} develops the high-dimensional model of task choice and directional technology adoption and derives the adoption condition, the adoption cone, and the intensive margin. Section \ref{sec:parametric} provides the parametric specification and its closed-form expressions. Section \ref{sec:discussion} relates the framework to prior task-based models and discusses interpretation and limitations. The Appendix collects proofs and additional derivations.

\section{A Model of Digital Technology Adoption}
\label{sec:model}

We study a worker doing a given job. We abstract from labor-market, governance, and principal-agent frictions, so the worker maximizes output. This corresponds to a richer environment where tasks are observable and contractible within the firm. We take the worker--job match as fixed and reserve the study of sorting for future work.

The worker chooses how much to do of each of \(N\) tasks. Denote the task quantities by the vector \(x\in\mathbb R_+^N\). The set of task bundles the worker can complete in a workday (or other unit of the scarce resource) is given by the set \(X_A\), characterized by its \emph{resource use function} \(g\), \[
X_A = \{x: g(x)\le B\}.
\] Here \(B\) is the amount of the scarce resource (``budget'') the worker has available. We interpret the budget as time or attention, not money. This is a key feature of digital technology adoption: the binding constraint is often the worker's cognitive capacity and hours, not financial cost. The set \(X_A\) is called the \emph{Production Possibilities Set} (PPS).

\begin{assumption}[Resource use function]
\label{ass:resource-use}
The \emph{resource use function} $g:\mathbb R^N_+\to\mathbb R_+$ is (i) convex, (ii) homogeneous of degree one, (iii) increasing in each component, and (iv) $g(0)=0$. These properties make the feasible set $X_A = \{x:g(x)\le B\}$ convex, scalable with the budget, and closed under free disposal. We do not require $g$ to be differentiable or strictly convex.
\end{assumption}

\begin{assumption}[Production function]
\label{ass:production}
The production function $F:\mathbb R^N_+\to\mathbb R_+$ is (i) homogeneous of degree one (constant returns to scale), (ii) strictly quasi-concave (upper contour sets $\{x:F(x)\ge c\}$ are strictly convex), and (iii) twice continuously differentiable on $\mathbb R^N_{++}$ with $\nabla F(x)\neq 0$ for all $x\in\mathbb R^N_{++}$. Strict quasi-concavity ensures that the optimum, when interior, is unique and the supporting hyperplane is pinned down by $\nabla F$.
\end{assumption}

\begin{definition}[Task allocation problem]
\label{def:task-allocation}
Given a resource use function $g$, a production function $F$, and a budget $B>0$, the worker's \emph{task allocation problem} is
$$
\max_{x\in\mathbb R^N_+} F(x)\qquad \text{s.t.}\qquad g(x)\le B.
$$
The worker chooses the task bundle that maximizes output subject to the constraint that total resource use does not exceed the budget.
\end{definition}
\begin{proposition}[Existence]
\label{prop:existence}
Under Assumptions \ref{ass:resource-use} and \ref{ass:production}, a solution to the task allocation problem (Definition \ref{def:task-allocation}) exists.
\end{proposition}
\begin{proof}
The feasible set $X_A$ is closed (preimage of $(-\infty,B]$ under continuous $g$) and bounded (homogeneity and monotonicity of $g$ imply $g(x)\to\infty$ as $\|x\|\to\infty$). The objective $F$ is continuous ($C^2$ implies continuity). Existence follows from Weierstrass's theorem.
\end{proof}

We denote any solution by \(x_A\).

\begin{definition}[Shadow prices]
\label{def:shadow-prices}
A \emph{shadow price} vector $p\in\mathbb R^N_+$ is the outward normal to a hyperplane supporting the Production Possibility Set at the optimum. We normalize all price vectors to have unit Euclidean length, $\|p\|_2 = 1$.
\end{definition}

\begin{proposition}[Autarky price characterization]
\label{prop:autarky-prices}
A task bundle $x_A$ solves the task allocation problem if and only if there exists an autarky price vector $p_A$ such that (i) $x_A$ maximizes revenue $R = p_A' x$ subject to $g(x)\le B$, and (ii) $x_A$ maximizes output $F(x)$ subject to $p_A' x \le R$. Moreover, if $x_A\in\mathbb R^N_{++}$, then $p_A$ is unique.
\end{proposition}
\begin{proof}
The first part follows from a standard tangency argument: at the optimum, a supporting hyperplane separates the feasible set from higher-output bundles, and its normal vector can be interpreted as implicit (autarky) task prices. When $x_A\in\mathbb R^N_{++}$, strict quasi-concavity of $F$ (Assumption \ref{ass:production}) implies that $\nabla F(x_A)$ uniquely pins down the supporting hyperplane, making $p_A$ unique even if $g$ is not differentiable at $x_A$.
\end{proof}

Figure \ref{fig:autarky-prices} shows optimal task allocation with smooth \(g\) and \(F\) functions. The worker chooses the task bundle \(x_A\) that maximizes output \(F(x)\) subject to the Production Possibility Set (PPS). The autarky price vector \(p_A\) is normal to the supporting hyperplane at the optimum. The boundary of the PPS is called the Production Possibility Frontier (PPF).

\begin{figure}[htbp]
\centering
\begin{tikzpicture}[scale=2.2]
\draw[->,thick,gray] (-0.1,0) -- (2.6,0) node[right] {\small Task 1};
\draw[->,thick,gray] (0,-0.1) -- (0,2.4) node[above] {\small Task 2};
\draw[thick,blue,domain=0:90,samples=80] plot ({2*cos(\x)},{1.6*sin(\x)});
\node[blue,left] at (0.3,1.45) {\small PPS};
\draw[thick,black!60,domain=0.75:2.4,samples=50] plot (\x,{1.590/\x});
\node[black!60,right] at (2.3,0.75) {\small $F(x)=\bar F$};
\fill[black] (1.491,1.067) circle (1.5pt);
\node[above right] at (1.491,1.067) {\small $x_A$};
\draw[dashed,black!50] (0,2.133) -- (2.5,0.345);
\node[black!50,below right] at (0.15,2.0) {\small $p_A' x = R$};
\draw[->,very thick,black] (1.491,1.067) -- (1.824,1.439);
\node[above] at (1.85,1.5) {\small $p_A$};
\end{tikzpicture}
\caption{Task allocation in autarky}
\label{fig:autarky-prices}
\end{figure}
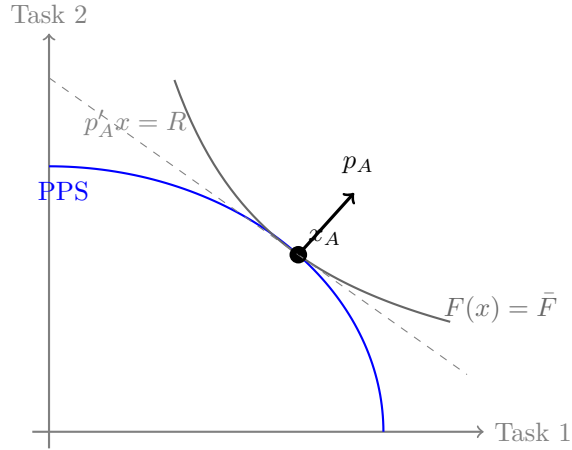

For notational convenience, we introduce the \emph{unit revenue} function as \[
\varrho(p) = \frac1B\max_{x:g(x)\le B} p'x = \max_{x:g(x)\le 1} p'x
\] that shows, for every possible set of prices, the maximum ``revenue'' the worker can achieve using only human activities for one unit of time (or resource).\footnote{Geometrically, this function represents a directional diameter of the PPS \(X_A\). While the extent of \(X_A\) is limited in each direction \(x\) by \(g(x)\le B\), this function converts the directional extent into an inner product.} Clearly, \(\varrho(p)\) is homogeneous of degree one and convex.

\subsection{Technology expands the production possibilities set}\label{technology-expands-the-production-possibilities-set}

We model AI as a directional technology capable of executing tasks in fixed proportions. Concretely, AI produces tasks in a fixed mix \(t\), scaled up by an intensity \(\lambda \ge 0\). We normalize \(\|t\|=1\) so that \(t\) only represents a \emph{direction}, not a length. Elements are \(t_i \ge 0\) but may be zero in any particular direction. For example, AI may not be able to help in manual tasks.

We first discuss the properties and the effects of a single technology. We generalize the technology adoption problem to multiple directional technologies in Section \ref{sec:multi-technology}.

To use AI, the worker has to use some of the same resources that constrain them in their own work. For example, instructing AI agents and reviewing their work takes time, and this limits the amount of AI that can be used in a workday.\footnote{Time as a binding constraint is central to the economics of digital goods. \citet{GoodmanList2026} leverage Becker's time allocation theory to study consumer demand for time-intensive platforms like Facebook and Instagram, finding that time shares and time costs significantly influence substitution patterns across online activities.} We assume AI is competing for the same resource budget \(B\), so if the worker spends a fraction \(\lambda\) of their budget supervising AI, only \((1-\lambda)B\) is left for doing other tasks.\footnote{For simplicity, we assume there is a linear transformation between resources spent on AI supervision and resources spent on human work. This is eminently true if the resource is time: one hour lost on an activity is one hour gained on another. A more general formulation could model the production possibilities set as constrained by several resources, such as time, attention, mental energy, and AI may not use them in the same proportion as human tasks. We conjecture our results would survive as long as AI is sufficiently substitutable to human tasks in the resource constraint.} We denote the \emph{capability} of technology with \(\chi\ge 0\), which acts as a productivity shifter. A unit of technology \(\|t\|\) costs \(1/\chi\) units of the resource \(B\). Continuing with the AI example, more capable AI tools can get more done per hour of human supervision.

\begin{definition}[Technology]
\label{def:technology}
A \emph{technology} is a pair $(t, \chi)$. The vector $t \in \mathbb{R}^N_+$ with $\|t\|_2 = 1$ is the \emph{direction} (which tasks the technology performs, and in what proportions). The scalar $\chi \ge 0$ is the \emph{capability} (how many units of tasks the technology produces per unit of resource spent supervising it).
\end{definition}

Allocating \(r\) units of resource to technology yields \(r \chi t\) units of tasks. When the worker allocates a fraction \(\lambda \in [0,1]\) of their budget \(B\) to technology, they obtain \(\lambda B \chi t\) from technology and have \((1-\lambda)B\) resources left for human work.

\begin{definition}[Technology Adoption Problem]
\label{def:tech-adoption}
Given technology $(t, \chi)$ and budget $B$, the worker chooses technology intensity $\lambda \in [0,1]$ and human task allocation $x_H$ to maximize output $F(x_H + \lambda B \chi t)$ subject to $g(x_H) \le (1-\lambda)B$.
\end{definition}

The Production Possibility Set with technology is the set of all task bundles the worker can achieve by combining human work with (possibly partial) use of technology, allowing free disposal of any tasks: \[
X_T = \left\{ z \in \mathbb R^N_+ : z \le x_H + \lambda B \chi t, \; g(x_H) \le (1-\lambda)B, \; \lambda \in [0,1] \right\}.
\] The inequality \(z \le x_H + \lambda B\chi t\) (componentwise) reflects free disposal: the worker can always discard output in any task. The autarky PPS \(X_A\) already satisfies free disposal (Assumption \ref{ass:resource-use}), but the technology contribution \(\lambda B\chi t\) is a single point in task space. Without the free disposal clause, \(X_T\) would not be downward-closed in the technology direction.

As can already be seen from the definition, technology expands the production possibilities set into a convex hull, as we state formally below. Figure \ref{fig:technology-tangent-segment} illustrates the geometry. The worker allocates a fraction \(\lambda\) of the budget to using technology (vector \(y\)) and a fraction \(1-\lambda\) to their own work. The optimal allocation lies on the tangent line from \(y\) to the Production Possibility Set. Point \(a\) represents own work, and the line segment from \(a\) to \(b\) represents the contribution of technology.

\begin{proposition}[Technology expands PPS to convex hull]
\label{prop:convex-hull}
The Production Possibility Set with technology adoption equals the convex hull of the autarky PPS and the technology point, closed under free disposal:
$$
X_T = \text{conv}(X_A \cup \{B\chi t\}) - \mathbb R^N_+,
$$
where the subtraction denotes the Minkowski difference (i.e., for any $z$ in the convex hull, all $z' \le z$ with $z'\ge 0$ are also in $X_T$). Since $X_A$ is already downward-closed, this simplifies to $X_T = \text{conv}(X_A \cup \{B\chi t\})$ whenever the technology point $B\chi t$ is dominated componentwise by some element of $X_A$, and otherwise augments the convex hull by the free-disposal region below the technology point.
\end{proposition}
\begin{proof}
Any $z \in X_T$ satisfies $z \le (1-\lambda)\tilde{x} + \lambda (B\chi t)$ for some $\tilde{x} = x_H/(1-\lambda) \in X_A$ (using homogeneity of $g$), hence $z$ lies in the downward closure of the convex hull. Conversely, take any $w = (1-\mu)x + \mu(B\chi t)$ with $x \in X_A$, $\mu \in [0,1]$, and any $z \le w$, $z \ge 0$. Set $\lambda = \mu$ and $x_H = (1-\mu)x$. Then $g(x_H) = (1-\mu)g(x) \le (1-\mu)B$ and $z \le x_H + \lambda B\chi t$, so $z \in X_T$.
\end{proof}

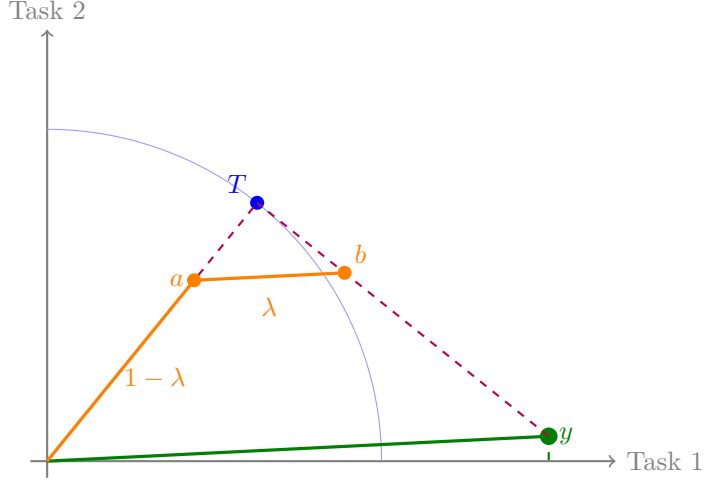
\begin{figure}[htbp]
\centering
\begin{tikzpicture}[scale=2.2]
\draw[->,thick,gray] (-0.1,0) -- (3.4,0) node[right] {\small Task 1};
\draw[->,thick,gray] (0,-0.1) -- (0,2.6) node[above] {\small Task 2};
\draw[thin,blue!40,domain=0:90,samples=80] plot ({2*cos(\x)},{2*sin(\x)});
\draw[very thick,green!50!black] (0,0) -- (3.0,0.15);
\draw[thick,dashed,green!50!black] (3.0,0.15) -- (3.0,0);
\fill[green!50!black] (3.0,0.15) circle (1.5pt);
\node[green!50!black,right] at (3.0,0.15) {\small $y$};
\fill[blue] (1.256,1.557) circle (1.2pt);
\node[blue,above left] at (1.256,1.557) {\small $T$};
\draw[thick,dashed,purple] (3.0,0.15) -- (1.256,1.557);
\draw[thick,dashed,purple] (0,0) -- (1.256,1.557);
\fill[orange] (0.879,1.090) circle (1.2pt);
\node[orange,left] at (0.879,1.090) {\small $a$};
\draw[very thick,orange] (0,0) -- (0.879,1.090);
\node[orange,right] at (0.40,0.50) {\small $1-\lambda$};
\fill[orange] (1.779,1.135) circle (1.2pt);
\node[orange,above right] at (1.779,1.135) {\small $b$};
\draw[very thick,orange] (0.879,1.090) -- (1.779,1.135);
\node[orange,below] at (1.329,1.05) {\small $\lambda$};
\end{tikzpicture}
\caption{Technology adoption as a convex combination}
\label{fig:technology-tangent-segment}
\end{figure}

The new PPS is larger than the previous one whenever the technology vector ``pierces'' the old PPS, that is, technology has an \emph{absolute advantage} over the worker in some combination of human tasks. When the entire technology vector is inside the worker's PPS, they do not gain anything from it, because the resource spent on working with the technology is less productive than doing the same combination of tasks manually.

Let \(x_t\) denote the combination of tasks that the worker could produce on their own in the exact direction of technology.

\begin{definition}[Technology-direction bundle]
\label{def:x-t}
For technology direction $t$, define $x_t$ as the point on the PPF in direction $t$:
$$
x_t = \frac{B}{g(t)} t.
$$
\end{definition}

This follows from homogeneity: \(g(\epsilon t) = \epsilon g(t)\), so the maximum \(\epsilon\) satisfying \(g(\epsilon t) \le B\) is \(\epsilon^* = B/g(t)\).

\begin{definition}[Absolute advantage]
\label{def:absolute-advantage}
Technology $(t, \chi)$ has \emph{absolute advantage} over a worker with resource use function $g$ if
$$
\chi g(t) > 1.
$$
\end{definition}

Equivalently, \(\|B\chi t\| > \|x_t\|\): technology can produce more in direction \(t\) than the worker. Since \(\|x_t\| = B/g(t)\) and \(\|B\chi t\| = B\chi\), the condition reduces to \(\chi > 1/g(t)\). Note that absolute advantage is independent of budget \(B\). Absolute advantage can arise either because technology is very capable (\(\chi\) large) or because the worker is particularly inefficient in direction \(t\) (large \(g(t)\), meaning high resource use per unit of tasks in that direction).

\begin{proposition}[Absolute advantage necessary for adoption]
\label{prop:absolute-advantage-necessary}
If technology $(t, \chi)$ does not have absolute advantage ($\chi g(t) \le 1$), then $X_T = X_A$: technology does not expand the PPS.
\end{proposition}
\begin{proof}
Without absolute advantage, $B\chi \le B/g(t)$, so $g(B\chi t) = B\chi g(t) \le B$. Thus the technology point $y = B\chi t \in X_A$. Since $X_A$ is convex, $\text{conv}(X_A \cup \{y\}) = X_A$.
\end{proof}

This proposition is illustrated in Figure \ref{fig:chi-less-than-one}. When the technology vector \(y = B\chi t\) does not extend beyond the worker's PPS, the worker does not adopt technology. The supporting hyperplane at the tangency point shows that the worker can achieve higher output using their own skills alone.\footnote{The figures normalize $B=1$, so the technology vector is shown as $y = \chi t$.}

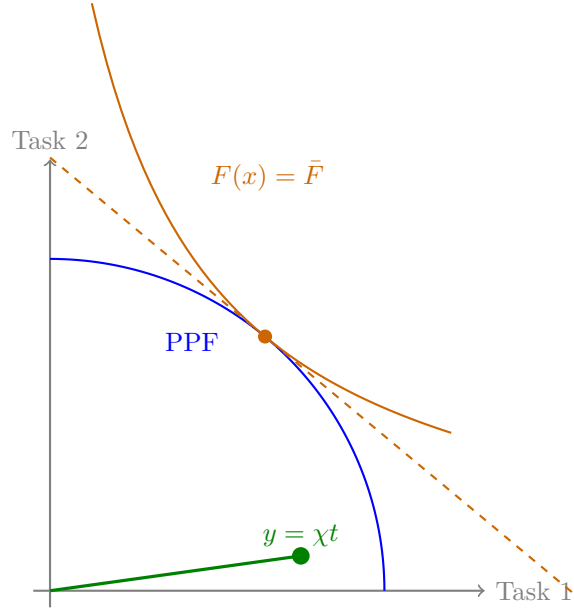
\begin{figure}[htbp]
\centering
\begin{tikzpicture}[scale=2.2]
\draw[->,thick,gray] (-0.1,0) -- (2.6,0) node[right] {\small Task 1};
\draw[->,thick,gray] (0,-0.1) -- (0,2.6) node[above] {\small Task 2};
\draw[thick,blue,domain=0:90,samples=80] plot ({2*cos(\x)},{2*sin(\x)});
\node[blue] at (0.85,1.5) {\small PPF};
\draw[thick,orange!80!black,domain=0.25:2.4,samples=50] plot (\x,{2.797/(\x+0.540)});
\node[orange!80!black,above right] at (0.9,2.35) {\small $F(x)=\bar F$};
\fill[orange!80!black] (1.286,1.532) circle (1.2pt);
\draw[thick,dashed,orange!80!black,domain=0:3.15] plot (\x,{-0.8391*\x+2.6108});
\draw[very thick,green!50!black] (0,0) -- (1.5,0.21);
\fill[green!50!black] (1.5,0.21) circle (1.5pt);
\node[green!50!black,above] at (1.5,0.21) {\small $y = \chi t$};
\end{tikzpicture}
\caption{Technology inside the Production Possibility Set ($\chi g(t) < 1$)}
\label{fig:chi-less-than-one}
\end{figure}

Absolute advantage is necessary but not sufficient for technology adoption. A counterexample is illustrated in Figure \ref{fig:chi-greater-than-one}. The technology already has an absolute advantage, so the worker \emph{could} use it to expand their PPS. In their current job, however, their activities are \emph{directionally far} from the technology, so it is not worth adopting it.\footnote{Geometrically, until the technology vector pierces the budget set defined by the hyperplane separating the PPS from the production function upper contour set, the convex hull of the technology point and the PPS remains strictly within the budget set and offers no revenue-increasing opportunities to the worker.}

\begin{figure}[htbp]
\centering
\begin{tikzpicture}[scale=2.2]
\draw[->,thick,gray] (-0.1,0) -- (2.6,0) node[right] {\small Task 1};
\draw[->,thick,gray] (0,-0.1) -- (0,2.6) node[above] {\small Task 2};
\draw[thick,blue,domain=0:90,samples=80] plot ({2*cos(\x)},{2*sin(\x)});
\node[blue] at (0.85,1.5) {\small PPF};
\draw[thick,orange!80!black,domain=0.25:2.4,samples=50] plot (\x,{2.797/(\x+0.540)});
\node[orange!80!black,above right] at (0.9,2.35) {\small $F(x)=\bar F$};
\fill[orange!80!black] (1.286,1.532) circle (1.2pt);
\draw[thick,dashed,orange!80!black,domain=0:3.15] plot (\x,{-0.8391*\x+2.6108});
\draw[very thick,green!50!black] (0,0) -- (2.5,0.35);
\fill[green!50!black] (2.5,0.35) circle (1.5pt);
\node[green!50!black,above] at (2.5,0.35) {\small $y = \chi t$};
\draw[thick,dashed,green!50!black] (2.5,0.35) -- (2.5,0);
\draw[thick,dashed,green!50!black] (2.5,0.35) -- (1.32,1.5);
\end{tikzpicture}
\caption{Technology has an absolute but not comparative advantage}
\label{fig:chi-greater-than-one}
\end{figure}
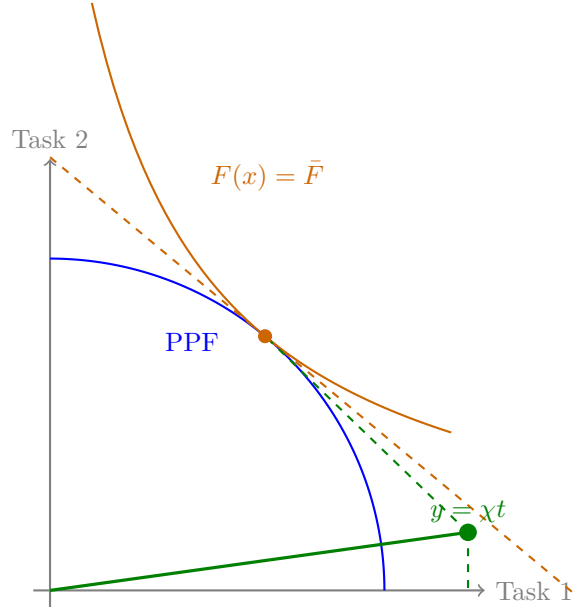

\subsection{The cone of technology adoption}\label{the-cone-of-technology-adoption}

Absolute advantage is necessary but not sufficient for adoption. The worker must also find technology \emph{directionally useful} given their job requirements. We now characterize when adoption occurs.

Recall that \(\varrho(p)\) is the unit revenue function, giving the implicit revenue the worker can achieve per unit of their resource when shadow prices are \(p\). The ``autarky revenue'' can then be written as \[
p_A'x_A = B\varrho(p_A).
\] For technology to be useful for the worker, it has to bring more implicit revenue.

\begin{definition}[Alignment angle]
\label{def:alignment-angle}
The \emph{alignment angle} $\phi$ between autarky prices $p_A$ and technology direction $t$ is
$$
\phi = \arccos(p_A' t) \in [0, \pi],
$$
where both $p_A$ and $t$ are unit vectors.
\end{definition}

The worker is indifferent to adopting technology when the technology point \(y = B\chi t\) brings exactly the same implicit revenue as the autarky activity vector. This happens when \(y\) lies exactly on the supporting hyperplane at the autarky allocation, that is, when \(p_A' y = p_A' x_A = B\varrho(p_A)\). This pins down a threshold capability for each technology direction such that the worker adopts that technology if and only if it is more capable than the threshold. Larger capability always leads to adoption: if the worker adopts at \(\chi\), the output-maximizing bundle under \(\chi\) remains feasible under any \(\chi' > \chi\), so the worker continues to adopt.

\begin{definition}[Adoption threshold]
\label{def:adoption-threshold}
For technology direction $t$ and shadow price $p$ with $p't > 0$, define the \emph{adoption threshold}
$$
\underline{\chi}(t,p) = \frac{\varrho(p)}{p't}.
$$
The adoption threshold is homogeneous of degree zero in prices and equals infinity when $p't = 0$. Because $p$ and $t$ are unit vectors, their inner product $p't = \cos\phi$ is the cosine of the alignment angle.
\end{definition}

The next theorem characterizes the adoption decision for a particular technology direction.

\begin{theorem}[Technology adoption condition]
\label{thm:adoption}
A worker with budget $B$, autarky prices $p_A$, and autarky allocation $x_A$ adopts technology $(t,\chi)$ if and only if
$$
B\chi\cdot p_A' t > p_A' x_A = B\varrho(p_A)
$$
or
$$
\chi \ge \chi_0 := \underline{\chi}(t,p_A).
$$
\end{theorem}
\begin{proof}
Technology is adopted if and only if the technology point $y = B\chi t$ lies on the profitable side of the supporting hyperplane through the autarky allocation $x_A$ with normal vector $p_A$. This is equivalent to $p_A' y > p_A' x_A$.
\end{proof}

The adoption condition does not depend on the overall scale of the budget \(B\): both the value of the technology point (\(p_A'\,B\chi t\)) and the value of the autarky allocation (\(p_A'x_A = B\varrho(p_A)\)) scale linearly in \(B\), so \(B\) cancels from the inequality. We interpret \(\chi_0\) as the capability threshold for adopting \emph{some} technology: at \(\chi = \chi_0\) the worker is indifferent, and for \(\chi\) just above \(\chi_0\) optimal adoption can be arbitrarily small (the worker is at the extensive margin with \(\lambda^* > 0\) but close to zero).

Since \(\cos\phi = p_A' t\), the adoption condition becomes \(\chi \cos\phi \ge \varrho(p_A)\). The condition has a team-production interpretation: the collaboration is worthwhile when the AI's contribution per unit of supervision (\(\chi\)), weighted by alignment (\(\cos\phi\)), exceeds the opportunity cost of supervision (\(\varrho(p_A)\)). This implies a geometric representation of the adoption decision. We can invert the problem and ask: for a given technology capability \(\chi\), what directions would the worker adopt?

\begin{theorem}[Cone of technology adoption]
\label{thm:cone-adoption}
For a worker with autarky prices $p_A$ and unit revenue function $\varrho$, the \emph{cone of technology adoption} for capability $\chi > \varrho(p_A)$ is the set of technology directions the worker adopts:
$$
\mathcal{C}(\chi, p_A, \varrho) = \left\{ t \in \mathbb R^{N} : p_A' t > \varrho(p_A)/\chi,\quad \|t\|_2=1 \right\}.
$$
The cone is centered on $p_A$ with half-angle $\phi_0 = \arccos[\varrho(p_A)/\chi]$. A worker adopts technology $(\chi, t)$ if and only if $t$ is within the cone of technology adoption, $t\in \mathcal C(\chi,p_A,\varrho)$. Equivalently, the alignment angle $\phi$ is less than $\phi_0$.
\end{theorem}

The cone has several important properties. First, it is non-empty exactly when \(\chi > \varrho(p_A)\). When \(\chi = \varrho(p_A)\), only the direction \(t = p_A\) is adopted. This provides a \emph{threshold condition} for technology adoption that is distinct from absolute advantage. Near the threshold, the cone widens rapidly. Formally, \(\partial\phi_0/\partial\chi\) is infinite at \(\chi = \varrho\). Finally, as technology capability increases without bound, \(\chi \to \infty\), the cone expands to a hemisphere (\(\phi_0 \to \pi/2\)). Figure \ref{fig:cone-3d} illustrates the cone geometry and Figure \ref{fig:cone-half-angle} shows how the half-angle varies with capability.

\begin{figure}[htbp]
\centering
\begin{tikzpicture}[scale=1.1]
  
  \newcommand{\gaugecircle}[2]{%
    \pgfmathsetmacro{\chiVal}{#1}%
    \pgfmathsetmacro{\rVal}{0.5*\B*sqrt(\chiVal*\chiVal - 1)}%
    \pgfmathsetmacro{\cx}{\chiVal*\B*\tx}%
    \pgfmathsetmacro{\cy}{\chiVal*\B*\ty}%
    \coordinate (center#2) at (\cx, \cy);
    \draw[dashed, rotate around={\tAng:(center#2)}, gray] 
      ($(center#2) + (0,{\rVal})$) arc (90:-90:{\rVal*0.25} and {\rVal});
    \draw[thick, rotate around={\tAng:(center#2)}, blue!70]
      ($(center#2) + (0,{\rVal})$) arc (90:270:{\rVal*0.25} and {\rVal});
  }
  
  \def\tAng{40}      
  
  \pgfmathsetmacro{\tx}{cos(\tAng)}
  \pgfmathsetmacro{\ty}{sin(\tAng)}
  
  \draw[->, thick, gray!70] (0,0) -- (5,0) node[right] {\small $x_1$};
  \draw[->, thick, gray!70] (0,0) -- (0,4.5) node[above] {\small $x_2$};
  \draw[->, thick, gray!70] (0,0) -- (-1.5,-2.6) node[below left] {\small $x_3$};
  
  \def\B{3}
  \shade[ball color=blue!15, opacity=0.5] (0,0) circle (\B);
  \draw[thick, blue!50] (0,0) circle (\B);
  
  \draw[thin, dashed, gray!60] ({\B*cos(25)}, {\B*sin(25)*0.25}) 
    arc (0:180:{\B*cos(25)} and {\B*cos(25)*0.2});
  \draw[thin, gray!60] ({\B*cos(25)}, {\B*sin(25)*0.25}) 
    arc (0:-90:{\B*cos(25)} and {\B*cos(25)*0.2});
  
  \coordinate (xt) at ({\B*\tx}, {\B*\ty});
  \fill[red] (xt) circle (2pt);
  \node[red, below left] at (xt) {\small $x_t$};
  
  \gaugecircle{1.05}{D}  
  \gaugecircle{1.15}{A}
  \gaugecircle{1.35}{B}
  \gaugecircle{1.7}{C}
  
  \pgfmathsetmacro{\ychi}{1.9}
  \coordinate (y) at ({\ychi*\B*\tx}, {\ychi*\B*\ty});
  \draw[green!50!black, dotted, thick] (0,0) -- (xt);
  \draw[->, thick, green!50!black] (xt) -- (y);
  \fill[green!50!black] (y) circle (2pt);
  \node[green!50!black, right] at (y) {\small $t$};
  
  \fill (0,0) circle (1.5pt);
  \node[below left] at (-0.2,-0.2) {\small $O$};
  
\end{tikzpicture}
\caption{Cone of technology adoption in 3D}
\label{fig:cone-3d}
\end{figure}
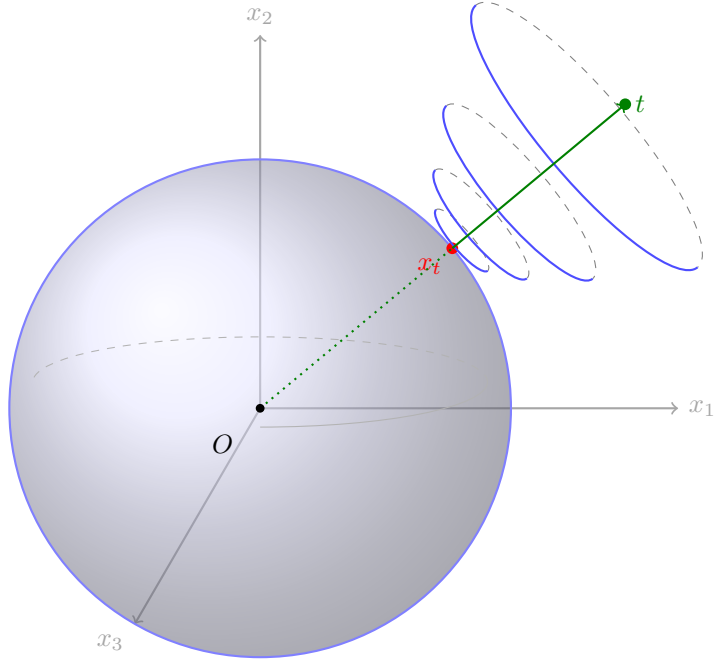

The unit sphere represents all possible technology directions \(t\). The cone is centered on the autarky price \(p_A\) with half-angle \(\phi_0 = \arccos(\varrho/\chi)\). The worker adopts any technology direction inside the cone. The dashed arcs show cross-sections of the cone at different capabilities \(\chi\), the cross-section has radius \(\sqrt{\chi^2-\varrho^2}\).

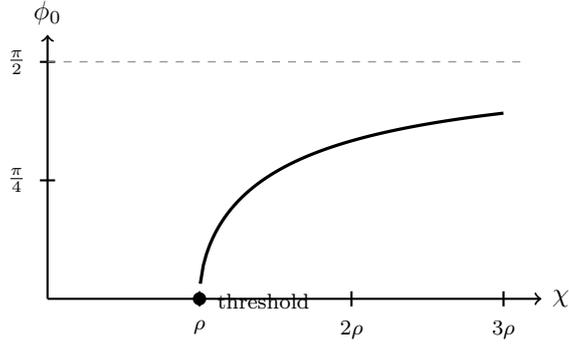
\begin{figure}[htbp]
\centering
\begin{tikzpicture}[scale=1.0]
\draw[->,thick] (0,0) -- (6.5,0) node[right] {\small $\chi$};
\draw[->,thick] (0,0) -- (0,3.5) node[above] {\small $\phi_0$};

\draw[thick] (-0.1,3.14159) -- (0.1,3.14159);
\node[left] at (-0.15,3.14159) {\scriptsize $\frac{\pi}{2}$};
\draw[thick] (-0.1,1.5708) -- (0.1,1.5708);
\node[left] at (-0.15,1.5708) {\scriptsize $\frac{\pi}{4}$};

\draw[thick] (2,-0.1) -- (2,0.1);
\node[below] at (2,-0.15) {\scriptsize $\rho$};
\draw[thick] (4,-0.1) -- (4,0.1);
\node[below] at (4,-0.15) {\scriptsize $2\rho$};
\draw[thick] (6,-0.1) -- (6,0.1);
\node[below] at (6,-0.15) {\scriptsize $3\rho$};

\draw[dashed,gray] (0,3.14159) -- (6.3,3.14159);

\draw[very thick,black,domain=2.01:6,samples=100] plot (\x,{acos(2/\x)*3.14159/90});

\fill[black] (2,0) circle (2.5pt);
\node[below right] at (2.1,0.2) {\scriptsize threshold};

\end{tikzpicture}
\caption{Cone half-angle as a function of technology capability}
\label{fig:cone-half-angle}
\end{figure}

\begin{corollary}[Comparative statics of the cone]
\label{cor:cone-comparative-statics}
Let $\varrho = \varrho(p_A)$ denote the worker's unit revenue at autarky. The cone of technology adoption has the following properties as a function of capability $\chi$:
\begin{enumerate}
\item \textbf{Threshold above absolute advantage.} The cone is non-empty if and only if $\chi > \varrho$. Since $\varrho \ge 1$ generically, with equality only when the worker's PPF passes through the unit simplex, adoption requires capability strictly above the absolute-advantage threshold $\chi = 1$.
\item \textbf{Fast widening near threshold.} As $\chi \to \varrho^+$, the cone half-angle satisfies
$$
\phi_0(\chi) \approx \sqrt{2\left(\frac{\chi}{\varrho} - 1\right)},
$$
so that $\phi_0$ grows with the square root of excess capability. In particular, $d\phi_0/d\chi \to +\infty$ as $\chi \to \varrho^+$: the range of adopted directions expands rapidly just above the threshold.
\item \textbf{Monotonic widening.} For all $\chi > \varrho$, $d\phi_0/d\chi > 0$. The cone half-angle increases monotonically from $0$ (at $\chi = \varrho$) to $\pi/2$ (as $\chi \to \infty$).
\end{enumerate}
\end{corollary}

Figure \ref{fig:cone-half-angle} illustrates these properties: the half-angle is zero at \(\chi = \varrho\), rises steeply just above the threshold (reflecting the square-root behavior), and increases monotonically toward \(\pi/2\) as capability grows. The steep initial slope means that small improvements in capability near the threshold translate into large expansions of the set of adopted directions---a prediction that could help explain rapid adoption cascades when a technology crosses from ``not quite good enough'' to ``just good enough.''

\subsection{The intensive margin of technology}\label{the-intensive-margin-of-technology}

The cone of technology adoption only describes that the worker uses \emph{some} technology, not how much. To understand this, we must also look at the production function \(F\).

The intensive margin is economically central because it determines whether technology is a marginal supplement to human work or a full replacement. In particular, whether adoption eventually reaches \(\lambda^* = 1\) governs how much of the worker's output is effectively pinned down by the technology direction rather than the worker's own comparative advantage, and it determines how strongly technology reshapes the job's implied shadow prices.

Let \(y=\chi B t\) the vector of activities that technology can produce when all the resources are spent on it. Optimal production is the solution to the problem \[
F(X_A,y) :=\max_{x,\lambda} F((1-\lambda)x + \lambda y)\qquad\text{s.t.}
\qquad x\in X_A.
\] Under what conditions will \(\lambda\) be interior, \(\lambda\in(0,1)\)?

\begin{proposition}[Unique optimal technology use]
\label{prop:optimal-tech-use}
Suppose $F$ is strictly quasi-concave and $C^2$. Let $y = B\chi t$ and $X_T = \mathrm{conv}(X_A \cup \{y\})$. Then:
\begin{enumerate}
\item The optimum $\tilde x_* = \arg\max_{\tilde x \in X_T} F(\tilde x)$ is unique.
\item There exist unique $x_* \in X_A$ and $\lambda_* \in [0,1]$ such that $\tilde x_* = (1-\lambda_*)x_* + \lambda_* y$.
\item The shadow price $p_* = \nabla F(\tilde x_*)/\|\nabla F(\tilde x_*)\|$ at the optimum is unique.
\end{enumerate}
\end{proposition}
\begin{proof}
Uniqueness of $\tilde x_*$ follows from strict quasi-concavity of $F$ over the convex set $X_T$. The decomposition follows from the convex hull representation (Proposition \ref{prop:convex-hull}). Uniqueness of $\lambda_*$ follows from strict concavity of $f(\lambda) = \max_{x \in X_A} F((1-\lambda)x + \lambda y)$ (which holds because $F$ is strictly quasi-concave and HD1), and $x_*$ is then pinned down by $x_* = (\tilde x_* - \lambda_* y)/(1-\lambda_*)$ when $\lambda_* < 1$. Uniqueness of $p_*$ follows from $C^2$ differentiability of $F$. See Appendix \ref{sec:proof-optimal-tech-use}.
\end{proof}

Take the optimal human task vector \(x_*\). Connect this point with the point \(y\) and study how the production function \(F\) varies along this segment. Define \[
f(\lambda) = F[(1-\lambda)x_* + \lambda y]\qquad\forall\lambda \in[0,1].
\] This univariate function is strictly concave in \(\lambda\) and maximized at \(\lambda_*\). The solution is interior if and only if the derivative of \(f\) is positive at zero and negative at one, \(f'(0) > 0\) and \(f'(1) < 0\).

The entry and corner conditions can be expressed in terms of the adoption threshold (Definition \ref{def:adoption-threshold}). See Appendix \ref{sec:proof-interior-adoption} for the full derivation. The entry threshold and the corner threshold are both instances of the adoption threshold (Definition \ref{def:adoption-threshold}), evaluated at different prices: \[
\chi_0 = \underline{\chi}(t, p_A), \qquad \chi_{100} = \underline{\chi}(t, p_{100}),
\] where \(p_{100} = \nabla F(t)/\|\nabla F(t)\|\) is the shadow price at the technology direction.

\begin{theorem}[Interior adoption]
\label{thm:interior-adoption}
Suppose $F$ is strictly quasi-concave, $C^2$, and homogeneous of degree one. Let $\lambda_*(\chi)$ denote the unique optimal technology intensity from Proposition \ref{prop:optimal-tech-use}. Then:
\begin{enumerate}
\item \emph{No adoption:} If $p_A't = 0$, then $\lambda_* = 0$ for all $\chi$: the worker never adopts.
\item \emph{Never all-in:} If $F(t) = 0$, then $\lambda_* < 1$ for all $\chi$: the worker never goes all in.
\item \emph{Partial adoption:} If $p_A't > 0$ and $F(t) > 0$, then $\chi_{100} > \chi_0$: partial adoption occurs for $\chi \in (\chi_0, \chi_{100})$, with $\lambda_* \in (0,1)$. For $\chi \ge \chi_{100}$, $\lambda_* = 1$ (all-in).
\end{enumerate}
\end{theorem}
\begin{proof}[Proof sketch]
Cases 1--2 follow from the boundary conditions of the univariate function $f(\lambda) = F[(1-\lambda)x_* + \lambda y]$. For case 3, the strict inequality $\chi_{100} > \chi_0$ follows from strict quasi-concavity: the supporting hyperplanes at $x_A$ and at the technology direction $t$ are distinct whenever $t$ is not collinear with $x_A$. See Appendix \ref{sec:proof-interior-adoption}.
\end{proof}

To compute these thresholds in practice, one proceeds in two steps. First, solve the autarky problem to obtain the autarky allocation \(x_A\) and its supporting normal \(p_A\) (under differentiability, \(p_A \propto \nabla g(x_A)\), normalized to unit length). Given \(p_A\), compute \(\chi_0 = \underline{\chi}(t,p_A) = \varrho(p_A)/(p_A't)\). Second, compute the technology-induced price \(p_{100}\) from the job's marginal rates of substitution at the technology direction (under differentiability, \(p_{100} \propto \nabla F(t)\), normalized). Then compute \(\chi_{100} = \underline{\chi}(t,p_{100})\). The gap \(\chi_{100} - \chi_0 > 0\) determines the width of the partial-adoption region.

Figure \ref{fig:interior} illustrates how to compute \(\chi_0\) and \(\chi_{100}\) in two dimensions. At capability \(\chi_0\) the technology vector is just sufficiently long to touch the budget line (hyperplane in \(N\) dimensions) spanned by the autarky prices at the production point \(x_A\). As capability grows, it will eventually reach a point where the production isoquant is parallel to the PPF at a new point \(x_{100}\). At this capability, production is maximized by using technology 100 percent of the time of the worker. As capability increases, this corner solution is going to remain, with the slope of the isoquant no longer pinning down technology usage.

The gap between \(\chi_{100}\) and \(\chi_0\) determines how long is the range of the ``intensive margin'' of technology usage. This depends on the curvature of both \(g\) and \(F\). For general curved PPFs, \(x_{100}\) is typically distinct from \(x_A\) and is ``directionally farther'' from the technology. This means that the last few human tasks that are done before full AI adoption are very different from what AI has to offer.

\begin{figure}[htbp]
\centering
\begin{tikzpicture}[scale=2.2]
\draw[->,thick,gray] (-0.1,0) -- (6.2,0) node[right] {\small Task 1};
\draw[->,thick,gray] (0,-0.1) -- (0,2.6) node[above] {\small Task 2};
\draw[thick,blue,domain=0:90,samples=80] plot ({2*cos(\x)},{2*sin(\x)});
\draw[thick,orange!80!black,domain=0.25:2.4,samples=50] plot (\x,{2.797/(\x+0.540)});
\draw[thick,orange!80!black,domain=3:6.2,samples=50] plot (\x,{4.8/(\x/1.45+0.540)-0.16});

\fill[orange!80!black] (1.286,1.532) circle (1.2pt);
\fill[blue] (0.89,1.79) circle (1.2pt);

\draw[thick,dashed,orange!80!black,domain=0:3.15] plot (\x,{-0.8391*\x+2.6108});

\draw[very thick,green!50!black] (0,0) -- (6.0,0.84);
\fill[green!50!black] (6.0,0.84) circle (1.5pt);
\fill[green!50!black] (2.7,0.378) circle (1.5pt);
\draw[thick,dashed,green!50!black] (6.0,0.84) -- (6.0,0);
\draw[thick,dashed,green!50!black] (6.0,0.84) -- (0.89,1.79);

\node[orange!80!black,above right] at (0.6,3.0) {\small $F(x)= F(x_A)$};
\node[orange!80!black,above] at (6.0,1.0) {\small $F(x)= F(\chi_{100}t)$};
\node[green!50!black,below right] at (2.7,0.378) {\small $\chi_{0} t$};
\node[green!50!black,below right] at (6.0,0.84) {\small $\chi_{100} t$};
\node[blue,below left] at (0.89,1.79) {\small $x_{100}$};
\node[orange!80!black,below left] at (1.286,1.532) {\small $x_{A}$};

\end{tikzpicture}
\caption{The limits of interior technology adoption}
\label{fig:interior}
\end{figure}
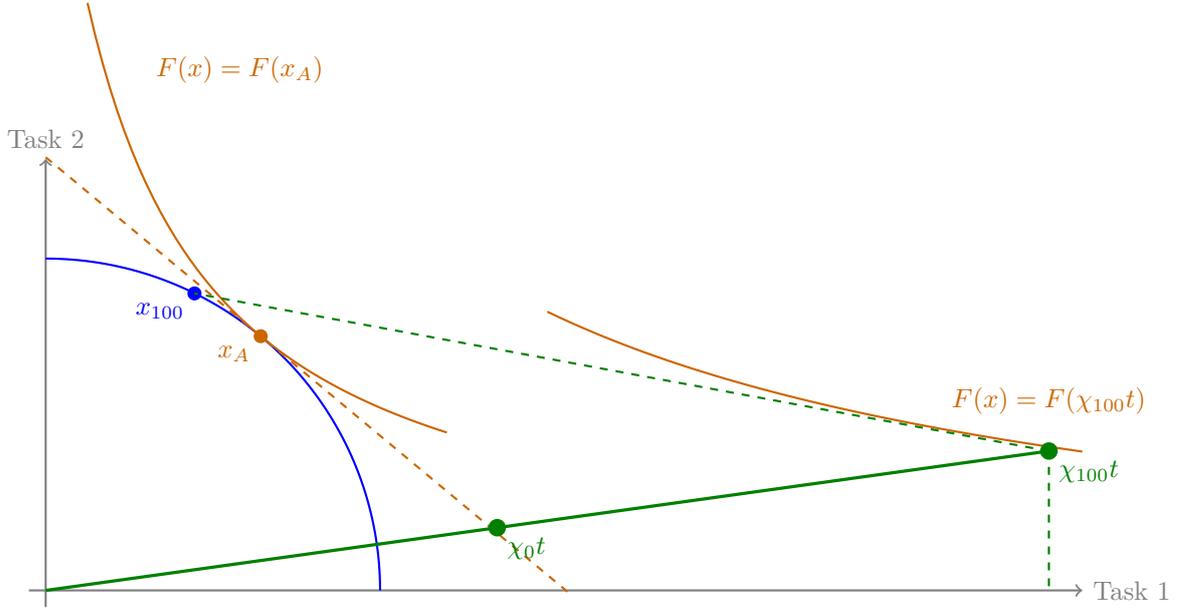

For technology adoption, the worker's skills and job requirements matter only through shadow prices. The autarky prices \(p_A\) summarize which task bundles are scarce at the margin in the current worker--job match: they encode the relevant interaction between the resource use function \(g\) (skills) and the production function \(F\) (requirements). The adoption condition in Theorem \ref{thm:adoption} depends on \(g\) and \(F\) only through \(p_A\) and the induced unit revenue \(\varrho(p_A)\), and the all-in threshold depends on \(F\) only through \(p_{100}\). As a result, two matches that generate the same \(p_A\) face the same cone of adopted technology directions for any given capability \(\chi\).

This sufficient-statistic property is the reason the worker--job match is the relevant unit of analysis and why one-dimensional exposure or skill rankings can misclassify adoption: what matters is the direction of shadow prices, not a scalar summary. We return to this interpretation in the Introduction.

\subsection{More than one technology}\label{more-than-one-technology}

\label{sec:multi-technology}

Workers rarely face a single tool in isolation: multiple technologies coexist, and adopting one reshapes the feasible set and the marginal value of tasks. This matters because adoption decisions are value comparisons at shadow prices. When a worker adopts an initial technology, the implied prices rotate, which can make a second technology more or less attractive even if its direction and capability are unchanged. This section extends the adoption and shadow-price logic to settings where several directional technologies are available simultaneously.

We proceed by induction on the number of technologies. The key insight is that after adopting \(K\) technologies, the expanded feasible set \(X_{T,K}\) retains the properties of the autarky PPS---convexity, compactness, free disposal---so the adoption logic for technology \(K+1\) is identical in structure to the single-technology case, with \(X_{T,K}\) replacing \(X_A\) and the current shadow prices \(p_K^*\) replacing \(p_A\).

\begin{definition}[Multi-technology PPS]
\label{def:multi-tech-pps}
Given technologies $(t_1,\chi_1),\ldots,(t_K,\chi_K)$ with technology points $y_k = B\chi_k t_k$, the \emph{$K$-technology PPS} is
$$
X_{T,K} = \left\{z \in \mathbb R^N_+ : z \le x_H + \sum_{k=1}^K \lambda_k y_k,\; g(x_H) \le \left(1-\sum_{k=1}^K \lambda_k\right)B,\; \lambda_k \ge 0,\; \sum_{k=1}^K \lambda_k \le 1\right\}.
$$
Under our assumptions on $g$, $X_{T,K} = \text{conv}(X_A \cup \{y_1,\ldots,y_K\})$ with free disposal (Proposition \ref{prop:convex-hull} applied inductively). See Appendix \ref{sec:proof-multi-tech}.
\end{definition}

\begin{proposition}[Inductive adoption structure]
\label{prop:inductive-adoption}
Let $\tilde x_K^*$ denote the optimum over $X_{T,K}$ and $p_K^* = \nabla F(\tilde x_K^*)/\|\nabla F(\tilde x_K^*)\|$ the associated shadow price. Then:
\begin{enumerate}
\item \textbf{Entry:} Technology $K+1$ with direction $t_{K+1}$ is adopted if and only if $\chi_{K+1} > \underline{\chi}(t_{K+1}, p_K^*)$, where $\underline{\chi}(t,p) = \varrho_K(p)/(p't)$ and $\varrho_K(p) = \max_{z \in X_{T,K}} p'z / B$ is the unit revenue over the $K$-technology PPS.
\item \textbf{Rising bar:} $\varrho_K(p) \ge \varrho_{K-1}(p)$ for all $p$ (the maximum over a larger set is weakly larger), with strict inequality for any $p$ such that one of $y_1,\ldots,y_K$ lies on the profitable side of the supporting hyperplane at $p$. Hence the entry threshold for technology $K+1$ is weakly higher than for technology $K$: each successive technology faces a tougher adoption bar.
\item \textbf{All-in:} The worker goes all-in on technology $K+1$ if $\chi_{K+1} > \varrho_K(p_{100,K+1})/p_{100,K+1}'t_{K+1}$ and $\chi_{K+1}/\chi_k > (p_{100,K+1}'t_k)/(p_{100,K+1}'t_{K+1})$ for all $k=1,\ldots,K$, where $p_{100,K+1} = \nabla F(t_{K+1})/\|\nabla F(t_{K+1})\|$.
\end{enumerate}
\end{proposition}
\begin{proof}
Part 1 follows from applying Theorem \ref{thm:adoption} with $X_{T,K}$ in place of $X_A$: the entry condition is $p_K^{*\prime} y_{K+1} > p_K^{*\prime} \tilde x_K^* = B\varrho_K(p_K^*)$. Part 2: $X_{T,K} \supseteq X_{T,K-1}$, so the support function (maximum over feasible bundles) is weakly larger; strictly so when the new technology point expands the set in direction $p$. Part 3 applies the all-in dominance logic of Theorem \ref{thm:interior-adoption} to $X_{T,K}$: at prices $p_{100,K+1}$, no element of $X_{T,K}$---neither human bundles nor previous technology points---delivers higher value. See Appendix \ref{sec:proof-all-in-K}.
\end{proof}

The all-in condition has a simple interpretation. Evaluate all feasible alternatives at the marginal value of output in the neighborhood of going all-in on the new tool, summarized by \(p_{100,K+1}\). The worker chooses \(y_{K+1}\) if, at these prices, the new technology delivers strictly higher value than any purely human bundle and than any previously available technology point. In other words, near \(y_{K+1}\) there is no competitive human or older-technology substitute.

\section{A parametric illustration}\label{sec:parametric}

This section provides a tractable parametric specialization of the general model. We model job requirements with a constant elasticity of substitution (CES) production function and skills with a constant elasticity of transformation (CET) resource use function. CES means that tasks are substitutable inputs in producing output, with elasticity parameter \(\sigma\). CET means that the worker can transform a fixed resource budget into different task bundles with constant elasticity \(\gamma\); it is a convenient parametric way to represent the Production Possibility Set.

The goal of the CES--CET specialization is to make the autarky shadow-price vector \(p^A\) operational by expressing it as a function of primitives: job requirements \(\theta\) and skills \(s\). Because adoption decisions depend on \(p^A\) (Theorems \ref{thm:adoption} and \ref{thm:cone-adoption}), this mapping turns the qualitative geometry of Section \ref{sec:model} into closed-form comparative statics.

\subsection{CES job requirements and CET skills}\label{ces-job-requirements-and-cet-skills}

\begin{definition}[CES--CET parametric environment]
\label{def:ces-cet}
The \emph{CES--CET environment} specializes the general model (Assumptions \ref{ass:resource-use} and \ref{ass:production}) as follows.

\textbf{Job requirements (CES).} The production function is constant elasticity of substitution (CES):
$$
F(x) = \left(\sum_{i=1}^N \theta_i^{1/\sigma}\, x_i^{(\sigma-1)/\sigma}\right)^{\sigma/(\sigma-1)},
$$
where $\theta = (\theta_1,\ldots,\theta_N)$ with $\theta_i > 0$ captures job requirements (task $i$'s value weight) and $\sigma > 0$ is the elasticity of substitution across tasks. $F$ satisfies constant returns, strict quasi-concavity, and $C^2$ differentiability on $\mathbb R^N_{++}$.

\textbf{Skills (CET).} The resource use function is constant elasticity of transformation (CET) \citep{PowellGruen1968}:
$$
g(x) = \left(\sum_{i=1}^N s_i^{-1/\gamma}\, x_i^{(\gamma+1)/\gamma}\right)^{\gamma/(\gamma+1)},
$$
where $s = (s_1,\ldots,s_N)$ with $s_i > 0$ captures skills (higher $s_i$ means lower resource cost for task $i$) and $\gamma > 0$ is the elasticity of transformation across tasks. $g$ is convex, homogeneous of degree one, increasing, and $g(0)=0$. See Appendix \ref{sec:ces-cet-derivations}.
\end{definition}

The parameters have straightforward economic content. The vector \(\theta\) captures job requirements: higher \(\theta_i\) means task \(i\) is more valuable for producing output. The vector \(s\) captures skills: higher \(s_i\) means task \(i\) uses fewer resources at the margin. The curvature parameters capture flexibility. A higher substitution elasticity \(\sigma\) makes tasks more interchangeable in production (a more flexible job), while a higher transformation elasticity \(\gamma\) makes it easier for the worker to reallocate effort across tasks (a more flexible skill set).

\subsection{Autarky allocation in closed form}\label{autarky-allocation-in-closed-form}

We begin with the autarky allocation problem (Definition \ref{def:task-allocation}). Under the CES--CET environment, the optimizer admits a closed-form expression. Define the CES--CET productivity index \[
\Phi(\theta,s) = \sum_{j=1}^N s_j^{(\sigma-1)/(\gamma+\sigma)}\, \theta_j^{(\gamma+1)/(\gamma+\sigma)}.
\] Then the autarky allocation takes the form \[
x_i^A = B\cdot \frac{s_i^{\sigma/(\gamma+\sigma)}\,\theta_i^{\gamma/(\gamma+\sigma)}}{\Phi^{\gamma/(\gamma+1)}},\qquad i=1,\ldots,N,
\] and autarky output per unit budget satisfies \[
\frac{Y^A}{B} = \frac{F(x^A)}{B} = \Phi^{\sigma/(\sigma-1)}.
\] See Appendix \ref{sec:ces-cet-derivations} for derivations.

The allocation formula shows that the worker concentrates activity where skills and requirements overlap. The exponents make the role of curvature transparent. When \(\gamma\) is high, the worker can freely reallocate across tasks, so they tilt toward what the job requires (higher weight on \(\theta_i\)). When \(\sigma\) is high, tasks are easily substitutable in production and requirements matter less, so the worker leans toward tasks where they are skilled (higher weight on \(s_i\)).

The productivity index \(\Phi\) has a simple alignment interpretation: output is high when the worker is relatively skilled in the tasks that the job values (roughly, when \(s_i\) is high where \(\theta_i\) is high). Misalignment lowers \(\Phi\) by forcing the worker to spend scarce resources on tasks that are both required and expensive to execute.

\subsection{Shadow prices as functions of job requirements and skills}\label{shadow-prices-as-functions-of-job-requirements-and-skills}

In the general model, autarky prices \(p^A\) are defined by tangency between the Production Possibility Frontier and an isoquant. Under CES--CET, they take a simple power-law form: \[
p_i^A \propto \left(\frac{\theta_i}{s_i}\right)^{1/(\gamma+\sigma)},\qquad \|p^A\|_2 = 1.
\] See Appendix \ref{sec:ces-cet-derivations}.

This formula clarifies the economic content of shadow prices. A task has high marginal value when it is important for output (high \(\theta_i\)) and costly to execute for the worker (low \(s_i\)). The relevant object is therefore the ratio \(\theta_i/s_i\): the gap between requirements and ease of execution.

The distinction between \(x^A\) and \(p^A\) is central for interpreting observed behavior. A worker may spend substantial time on a task either because it is required (high \(\theta_i\)) or because it is easy (high \(s_i\)). Shadow prices separate these forces: \(p_i^A\) is high precisely when the job demands the task but the worker finds it expensive. This is the key statistic for valuing a directional technology.

\subsection{Comparative statics inside the job (Jevons Paradox)}\label{comparative-statics-inside-the-job-jevons-paradox}

The closed forms make it easy to describe within-job reallocation. Define the autarky time (or activity) share of task \(i\) as \[
\omega_i = \frac{s_i^{-1/\gamma}x_i^{1+1/\gamma}}{\sum_{j=1}^N s_j^{-1/\gamma}x_j^{1+1/\gamma}}.
\] Substituting the allocation formula above, all common scaling terms cancel and we obtain \[
\omega_i^A
= \frac{s_i^{(\sigma-1)/(\gamma+\sigma)}\,\theta_i^{(\gamma+1)/(\gamma+\sigma)}}{\Phi}.
\]

An increase in skill \(s_i\) raises \(\omega_i^A\) if and only if the tasks are substitutes in the production function, \(\sigma>1\). This is Jevons' Paradox for knowledge work. Someone skilled at a task will spend \emph{more time} with it, not \emph{less}, when demand for that task is elastic. By contrast, when tasks are complement, each task has to be completed to some predefined level, and skilled workers spend less time on satisfying the job requirement. This also highlights that activities, in general, are poor measures of skills.

\subsection{Applying the adoption theorems under CES--CET}\label{applying-the-adoption-theorems-under-cescet}

We now express the key statements about technology adoption in this parametric environment. Theorems \ref{thm:adoption} and \ref{thm:cone-adoption} show that adoption depends on (i) autarky prices \(p^A\) (through alignment \(p^A\cdot t\)) and (ii) the unit revenue function \(\varrho(p^A)\). Under CES--CET, both objects can be expressed in closed form as functions of \((\theta,s,\sigma,\gamma)\).

\begin{proposition}[Adoption objects under CES--CET]
\label{prop:ces-cet-adoption}
Under the CES--CET environment (Definition \ref{def:ces-cet}), let $\Phi = \sum_j s_j^{(\sigma-1)/(\gamma+\sigma)} \theta_j^{(\gamma+1)/(\gamma+\sigma)}$:
\begin{enumerate}
\item \textbf{Unit revenue.} The unit revenue function is the dual of the CET:
$$
\varrho(p) = \left(\sum_{i=1}^N s_i\, p_i^{\gamma+1}\right)^{1/(\gamma+1)}.
$$
\item \textbf{Entry threshold.} In terms of primitives $(\theta, s, \sigma, \gamma, t)$:
$$
\chi_0(t) = \frac{\Phi^{1/(\gamma+1)}}{\sum_i t_i (\theta_i/s_i)^{1/(\gamma+\sigma)}}.
$$
The numerator $\Phi^{1/(\gamma+1)}$ measures the worker's autarky productivity: it is high when skills align with requirements. The denominator measures alignment of autarky prices with the technology direction: it is high when $t_i$ loads on tasks where $\theta_i/s_i$ is large (required but hard). Better-aligned technology faces a lower entry bar. 
\item \textbf{Corner threshold.} The corner threshold can be written as
$$
\chi_{100}(t) = \frac{\left(\sum_i s_i\, \theta_i^{(\gamma+1)/\sigma}\, t_i^{-(\gamma+1)/\sigma}\right)^{1/(\gamma+1)}}{\sum_i \theta_i^{1/\sigma}\, t_i^{1-1/\sigma}}.
$$
The numerator depends on $(s, \theta, \sigma, \gamma, t)$: the worker's unit revenue at technology-induced prices. The denominator depends on $(\theta, \sigma, t)$ only: how much value the technology direction generates at its own induced prices. Skills $s$ enter only the numerator: a more skilled worker has higher $\chi_{100}$ (harder to push to all-in).
\end{enumerate}
See Appendix \ref{sec:ces-cet-derivations} for derivations.
\end{proposition}

The next corollary derives some comparative statics of the adoption thresholds with respect to the curvature of the PPF and the production function. The first statement says that when jobs or skills are more flexible, even a less aligned technology can be productive. Adoption is faster in flexible jobs like general knowledge work (e.g., software, consulting, academia, journalism), and slower in very rigid jobs requiring highly specialized skills (e.g.~medicine) or exact satisfaction of job requirements (e.g., law, medicine, nuclear physics, aviation). The second statement says that flexible jobs not only adopt early, but are also more likely to go ``all in.'' For more rigid jobs, we will see a wider range of interior AI usage with \(\lambda\in(0,1)\).

\begin{corollary}[Curvature and adoption]
\label{cor:curvature-adoption}
Fix autarky productivity $\Phi$, job requirements $\theta$, skills $s$, and technology capability $\chi$. Then:
\begin{enumerate}
\item \emph{Cone width increases with flexibility.} The cone of adopted technology directions is wider when $\gamma + \sigma$ is large. In the limit $\gamma + \sigma \to \infty$, the cone converges to a hemisphere (all directions with positive alignment are adopted).
\item \emph{Partial-adoption region narrows with flexibility.} The ratio $\chi_{100}/\chi_0$ is decreasing in $\gamma$ and in $\sigma$. In the limit, adoption becomes all-or-nothing: either $\lambda_* = 0$ or $\lambda_* = 1$.
\end{enumerate}
\end{corollary}
\begin{proof}[Proof sketch]
(1) The autarky price satisfies $p_i^A \propto (\theta_i/s_i)^{1/(\gamma+\sigma)}$. As $\gamma + \sigma \to \infty$, the exponent $1/(\gamma+\sigma) \to 0$, so $p^A$ converges to the uniform direction. A more uniform price vector is positively aligned with a larger set of technology directions, widening the cone. (2) The thresholds $\chi_0$ and $\chi_{100}$ are evaluated at $p_A$ and $p_{100}$ respectively. Higher curvature compresses both price vectors toward uniformity, reducing the gap between them. See Appendix \ref{sec:proof-curvature-adoption}.
\end{proof}

\section{Discussion and Relation to the Literature}
\label{sec:discussion}

Task-based labor-market models have a long history. In contrast to much of this literature, and in line with international trade theory \citep{DixitNorman1980,HelpmanKrugman1985}, we model what workers can do as a convex production possibility set. This allows workers to flexibly adjust the tasks they supply in response not only to job requirements but also to technical change. Here we discuss how our approach differs from earlier models and how these differences affect our understanding of AI adoption.

Our framework rests on two key assumptions: \emph{task bundling} and \emph{worker flexibility}.

\textit{Task bundling} refers to models where jobs require multiple tasks that are not tradable outside the specific worker-firm relationship. Foundational skill-bundle models \citep{Welch1969,KatzMurphy1992,Lazear2009,AutorThompson2025} treat workers as endowed with multiple skills that firms combine with heterogeneous weights. Sorting models \citep{GathmannSchoenberg2010,LisePostelVinay2020,BaranyBuchinskyPapp2020,BaranyHolzheu2026} study how workers with multidimensional skills match to jobs with multidimensional requirements. \citet{Acemoglu2011,FortinLemieux2016,CavounidisLang2020,BenzellMyers2026} use such task bundling models to study recent waves of automation.

\textit{Worker flexibility} refers to whether workers can reallocate effort across tasks within a job. Models with flexibility include the routine-biased literature \citep{Autor2003,Acemoglu2011,Cortes2016,Bohm2020}, the automation-threshold models \citep{AcemogluRestrepo2018,AcemogluRestrepo2019,ChenLiang2026}, expertise compression \citep{AutorThompson2025}, and most recent AI papers \citep{Freund2025,AlthoffReichardt2026,GansGoldfarb2026,CavounidisEtAl2024,CavounidisLang2020}. Most models assume perfectly elastic reallocation across tasks, which, in our model, corresponds to \(\gamma=\infty\). This assumption, when tasks are not bundled results in market task prices, which are neither worker- nor job-specific. If tasks are bundled, each worker's shadow prices are pinned down by their relative skills rather than by match-specific marginal valuations, and thus are not job-specific.

Very few models work with finite worker flexibility among a set of bundled tasks. Among the papers we survey, only \citet{CavounidisLang2020} allow workers to reallocate across tasks at a positive but finite cost, leading to the same CET production possibilities set as in our model. Their focus, however, is dynamic skill formation rather than technology adoption.

Our paper gives a general characterization of the technology adoption problem for any degree of worker and job flexibility, thereby naturally embedding endowment models (\(\gamma=0\)), market-based task models (\(\sigma=\infty\)) and any intermediate case. To the best of our knowledge, we are also the first to model AI adoption as a team production problem that makes the combined feasible set convex.

\section{Conclusion}\label{conclusion}

We model AI adoption as collaboration under a shared time and attention budget. A worker's feasible task combinations form a convex production possibility set; an AI tool adds another such set with its own directional capabilities. Because using the tool requires supervision, the relevant object is the weighted sum of the worker and tool sets, with weights determined by time allocation. Adoption is therefore neither pure substitution nor pure augmentation: it is a team-production choice with an endogenous intensity.

This framing sharpens what can (and cannot) be predicted without context. The right sufficient statistic is the worker's autarky (shadow) prices at their pre-AI optimum: they summarize what is scarce at the margin in the worker--job match and translate the tool's multidimensional direction into a scalar value. Adoption occurs when the tool's direction is valuable at these prices. Because shadow prices are about marginal valuations, they need not coincide with observed task allocations, skills, or job requirements taken separately.

Two implications follow for empirical work on GenAI in knowledge jobs. First, match-specificity is central: neither ``which workers adopt'' nor ``which jobs are exposed'' is well-defined in isolation because the same tool can be valuable in one match and irrelevant in another. Second, adoption has structured margins. As capabilities improve, adoption expands along an extensive margin (the set of adopted directions) and an intensive margin (partial versus all-in use), so small capability improvements near the threshold can generate large changes in usage while leaving other matches unchanged. These features help explain why scalar exposure indices can misclassify both adoption and gains when the technological frontier is jagged.


\bibliographystyle{agsm}
\bibliography{references}


\appendix

\section{Proofs and Derivations}\label{sec:appendix}

\subsection{\texorpdfstring{Proof of Proposition \ref{prop:optimal-tech-use}}{Proof of Proposition }}\label{sec:proof-optimal-tech-use}

\emph{Uniqueness of \(\tilde x_*\).} Suppose \(\tilde x\) and \(\tilde x'\) are both optimal with \(\tilde x \neq \tilde x'\). Since \(X_T\) is convex, \(\bar x = \tfrac{1}{2}\tilde x + \tfrac{1}{2}\tilde x' \in X_T\). By strict concavity of \(F\), \(F(\bar x) > \tfrac{1}{2}F(\tilde x) + \tfrac{1}{2}F(\tilde x') = F(\tilde x)\), contradicting optimality. Hence \(\tilde x_*\) is unique.

\emph{Uniqueness of \(\lambda_*\).} Define \(h(\lambda) = \max_{x \in X_A} F((1-\lambda)x + \lambda y)\). For each fixed \(x \in X_A\), \(\lambda \mapsto F((1-\lambda)x + \lambda y)\) is strictly concave (composition of strictly concave \(F\) with an affine function of \(\lambda\)). The function \(h(\lambda)\) is strictly concave by a standard parametric optimization argument: take \(\lambda_1 \neq \lambda_2\) with maximizers \(x_1^*, x_2^*\); the two optimal bundles \((1-\lambda_i)x_i^* + \lambda_i y\) are distinct (since \(\tilde x_*\) is unique and varies with \(\lambda\)), so strict concavity of \(F\) gives \(h(\bar\lambda) > \mu h(\lambda_1) + (1-\mu)h(\lambda_2)\). Hence \(\lambda_*\) is unique, and \(x_* = (\tilde x_* - \lambda_* y)/(1-\lambda_*)\) when \(\lambda_* < 1\).

\emph{Uniqueness of \(p_T\).} Since \(F\) is \(C^2\), \(\nabla F(\tilde x_*)\) is well-defined and nonzero. The shadow price \(p_T = \nabla F(\tilde x_*)/\|\nabla F(\tilde x_*)\|\) is therefore unique. \(\square\)

\subsection{\texorpdfstring{Proof of Theorem \ref{thm:interior-adoption}}{Proof of Theorem }}\label{sec:proof-interior-adoption}

Define the univariate function \(f(\lambda) = F[(1-\lambda)x_* + \lambda y]\), which is strictly concave in \(\lambda\) on \([0,1]\).

\emph{Case 1.} If \(p_A't = 0\), then \(f'(0) = \nabla F(x_A)'(y - x_A) = B\chi \nabla F(x_A)'t - \nabla F(x_A)'x_A\). Since \(\nabla F(x_A) \propto p_A\) and \(p_A't = 0\), we have \(\nabla F(x_A)'t = 0\), so \(f'(0) < 0\) and \(\lambda_* = 0\).

\emph{Case 2.} If \(F(t) = 0\), then \(F(y) = F(B\chi t) = B\chi F(t) = 0\) by homogeneity. Going all-in yields zero output, while autarky gives positive output \(F(x_A) > 0\). Hence \(\lambda_* < 1\).

\emph{Case 3.} When \(p_A't > 0\) and \(F(t) > 0\), both \(\chi_0 = \underline{\chi}(t, p_A)\) and \(\chi_{100} = \underline{\chi}(t, p_{100})\) are finite. The entry condition \(f'(0) > 0\) holds iff \(\chi > \chi_0\), and the corner condition \(f'(1) < 0\) holds iff \(\chi < \chi_{100}\). By Lemma \ref{lem:threshold-ordering} below, \(\chi_{100} > \chi_0\) whenever \(t\) is not collinear with \(x_A\). For \(\chi \in (\chi_0, \chi_{100})\), we have \(f'(0) > 0\) and \(f'(1) < 0\), so by strict concavity \(\lambda_* \in (0,1)\). For \(\chi \ge \chi_{100}\), \(f'(1) \ge 0\), so \(\lambda_* = 1\). \(\square\)

\begin{lemma}[Threshold ordering]
\label{lem:threshold-ordering}
If $F$ is strictly concave, homogeneous of degree one, and the technology direction $t$ is not collinear with the autarky allocation $x_A$, then $\chi_{100} > \chi_0$.
\end{lemma}
\begin{proof}
Define $c = F(x_A)/F(t)$, so that $ct$ lies on the same isoquant as $x_A$. By strict concavity, the supporting hyperplane at $x_A$ lies strictly above $F(t)$: $F(t) < \nabla F(x_A)'t$. Using Euler's theorem $F(x_A) = \nabla F(x_A)'x_A$ and the definition $\chi_0 = F(x_A)/(\nabla F(x_A)'t)$, we have $c > \chi_0$. Similarly, the supporting hyperplane at $t$ lies strictly above $F(x_A)$: $F(x_A) < \nabla F(t)'x_A$. Since $\chi_{100} \ge (\nabla F(t)'x_A)/F(t)$ and $\nabla F(t)'x_A > F(x_A)$, we have $\chi_{100} > c$. Combining: $\chi_0 < c < \chi_{100}$. $\square$
\end{proof}

\subsection{Multi-technology PPS}\label{sec:proof-multi-tech}

\begin{proposition}
\label{prop:multi-tech-convex-hull}
Under our assumptions on $g$ (convex, homogeneous of degree 1), $X_{T,K} = \mathrm{conv}(X_A \cup \{y_1, \ldots, y_K\})$.
\end{proposition}
\begin{proof}
*$X_{T,K} \subseteq \mathrm{conv}(X_A \cup \{y_1, \ldots, y_K\})$:* Take $z = x_H + \sum_k \lambda_k y_k \in X_{T,K}$ with $\Lambda = \sum_k \lambda_k$. For $\Lambda < 1$, define $\tilde x = x_H/(1-\Lambda)$. By homogeneity, $g(\tilde x) = g(x_H)/(1-\Lambda) \le B$, so $\tilde x \in X_A$. Then $z = (1-\Lambda)\tilde x + \sum_k \lambda_k y_k$ is a convex combination.

*$\mathrm{conv}(X_A \cup \{y_1, \ldots, y_K\}) \subseteq X_{T,K}$:* Take $z = \mu_0 x + \sum_k \mu_k y_k$ with $x \in X_A$. Set $\lambda_k = \mu_k$ and $x_H = \mu_0 x$. Then $g(x_H) = \mu_0 g(x) \le \mu_0 B = (1-\sum_k \lambda_k)B$. $\square$
\end{proof}

\begin{lemma}[Unique shadow prices]
\label{lem:shadow-prices-multi}
If $F$ is strictly concave and $C^1$, the optimum $x^*_K = \arg\max_{x \in X_{T,K}} F(x)$ is unique and $p^*_K = \nabla F(x^*_K)/\|\nabla F(x^*_K)\|$ is well-defined.
\end{lemma}
\begin{proof}
Uniqueness of $x^*_K$ follows from strict concavity of $F$ over the convex set $X_{T,K}$. Continuous differentiability of $F$ ensures $\nabla F(x^*_K)$ exists and is nonzero, giving a unique $p^*_K$. $\square$
\end{proof}

\subsection{\texorpdfstring{All-in adoption of technology \(K+1\)}{All-in adoption of technology K+1}}\label{sec:proof-all-in-K}

\begin{proposition}
\label{prop:all-in-K+1-appendix}
Under Assumption \ref{ass:production}, the worker goes all-in on technology $K+1$ if
$$
p_{100,K+1}'y_{K+1} > \max\left\{ \max_{x \in X_A} p_{100,K+1}'x, \; p_{100,K+1}'y_1, \ldots, p_{100,K+1}'y_K \right\}.
$$
\end{proposition}
\begin{proof}
At $x^* = y_{K+1}$, the first-order condition requires $\nabla F(y_{K+1})'(x - y_{K+1}) \le 0$ for all $x \in X_{T,K+1}$. Normalizing and using the fact that any $x \in X_{T,K+1}$ is a convex combination of points in $X_A$ and technology points $\{y_1, \ldots, y_{K+1}\}$, the condition reduces to: $p_{100,K+1}'y_{K+1}$ exceeds $p_{100,K+1}'x$ for all extreme points $x$ of $X_{T,K+1}$. With strict inequality, $y_{K+1}$ is the unique optimum. $\square$
\end{proof}

\subsection{CES-CET closed forms}\label{sec:ces-cet-derivations}

The closed-form expressions stated in Section \ref{sec:model} are derived here.

\textbf{CES production function.} The function \(F(x;\theta,\sigma) = \left(\sum_i \theta_i^{1/\sigma} x_i^{(\sigma-1)/\sigma}\right)^{\sigma/(\sigma-1)}\) satisfies constant returns to scale, concavity for \(\sigma > 0\), and \(C^\infty\) differentiability on \(\mathbb R^N_{++}\).

\textbf{CET resource use function.} The function \(g(x;s,\gamma) = \left(\sum_i s_i^{-1/\gamma} x_i^{(\gamma+1)/\gamma}\right)^{\gamma/(\gamma+1)}\) is convex, homogeneous of degree 1, and increasing. These properties follow from the convexity of \(x^{(\gamma+1)/\gamma}\) for \(\gamma > 0\) and the composition rules for convex functions.

\textbf{Autarky allocation.} Setting up the Lagrangian \(\mathcal L = F(x) - \mu(g(x) - B)\) and taking the ratio of first-order conditions for tasks \(i\) and \(j\): \[
\frac{x_i}{x_j} = \left(\frac{\theta_i}{\theta_j}\right)^{\gamma/(\gamma+\sigma)} \left(\frac{s_i}{s_j}\right)^{\sigma/(\gamma+\sigma)}.
\] Substituting into the constraint \(g(x) = B\) yields the normalization \(\Phi = \sum_j s_j^{(\sigma-1)/(\gamma+\sigma)} \theta_j^{(\gamma+1)/(\gamma+\sigma)}\), giving the closed-form allocation in the text.

\textbf{Productivity index.} By constant returns to scale, \(F(x^A)/B = F(x^A/B)\). Substituting the allocation formula and simplifying yields \(F(x^A)/B = \Phi^{\sigma/(\sigma-1)}\).

\textbf{Shadow prices.} From the gradient of \(g\) at \(x^A\): \(p_i^A \propto s_i^{-1/\gamma}(x_i^A)^{1/\gamma}\). Substituting the allocation formula and simplifying: \(p_i^A \propto (\theta_i/s_i)^{1/(\gamma+\sigma)}\), normalized to \(\|p^A\|_2 = 1\).

\textbf{Time shares.} From the allocation formula, \(\omega_i^A = x_i^A / \sum_j x_j^A \propto s_i^{\sigma/(\gamma+\sigma)} \theta_i^{\gamma/(\gamma+\sigma)}\). Logarithmic differentiation gives \(\partial \omega_i^A/\partial s_i = [\sigma/((\gamma+\sigma)s_i)] \omega_i^A(1-\omega_i^A) > 0\).

\textbf{Unit revenue.} The unit revenue function for CET is \(\varrho(p) = (\sum_i s_i\, p_i^{\gamma+1})^{1/(\gamma+1)}\) (the dual \(L^{q'}\) norm with \(q' = \gamma+1\)).

\textbf{Entry threshold.} From Definition \ref{def:adoption-threshold}, \(\chi_0(t) = \underline{\chi}(t, p_A) = \varrho(p_A)/(p_A't)\). Substituting \(p_i^A = \xi_i/\|\xi\|\) with \(\xi_i = (\theta_i/s_i)^{1/(\gamma+\sigma)}\), the \(\|\xi\|\) terms cancel, giving \(\chi_0(t) = \Phi^{1/(\gamma+1)}/\sum_i \xi_i t_i\) where \(\Phi = \sum_j s_j^{(\sigma-1)/(\gamma+\sigma)} \theta_j^{(\gamma+1)/(\gamma+\sigma)}\).

\textbf{Corner threshold.} From Definition \ref{def:adoption-threshold}, \(\chi_{100}(t) = \underline{\chi}(t, p_{100})\) where \(p_{100} = \nabla F(t)/\|\nabla F(t)\|\). Under CES, \(\nabla_i F(x) \propto \theta_i^{1/\sigma} x_i^{-1/\sigma}\), so \(p_{100,i} \propto \theta_i^{1/\sigma} t_i^{-1/\sigma}\). Substituting into \(\varrho(p_{100})\) and \(p_{100}'t\): \[
\chi_{100}(t) = \frac{\left(\sum_i s_i\, \theta_i^{(\gamma+1)/\sigma}\, t_i^{-(\gamma+1)/\sigma}\right)^{1/(\gamma+1)}}{\sum_i \theta_i^{1/\sigma}\, t_i^{1-1/\sigma}}.
\] The numerator is the worker's unit revenue at technology-induced prices; the denominator is the technology's value at those prices. Skills \(s\) enter only the numerator: a more skilled worker has higher \(\chi_{100}\) (harder to push to all-in).

\textbf{Comparative statics.} Both \(\gamma\) and \(\sigma\) enter \(\chi_0\) through \(\xi_i = (\theta_i/s_i)^{1/(\gamma+\sigma)}\) and through the outer exponent \(1/(\gamma+1)\). Higher \(\gamma + \sigma\) compresses \(\xi_i\) toward uniformity, making autarky prices less dispersed. The net effect on \(\chi_0\) depends on the technology direction: technologies aimed at the worker's weak tasks (high \(\theta_i/s_i\)) see the denominator \(\sum_i \xi_i t_i\) fall, potentially raising \(\chi_0\).

\subsection{\texorpdfstring{Proof of Corollary \ref{cor:curvature-adoption}}{Proof of Corollary }}\label{sec:proof-curvature-adoption}

\textbf{(1) Cone width increases with \(\gamma + \sigma\).}

The cone of adopted directions at capability \(\chi\) is \(\mathcal C(\chi) = \{t \in \mathbb S^{N-1}_+ : p_A't > \varrho(p_A)/\chi\}\). The width of this cone depends on how concentrated \(p_A\) is: a more uniform \(p_A\) is positively aligned with more directions \(t\).

Under CES-CET, \(p_i^A \propto (\theta_i/s_i)^{1/(\gamma+\sigma)}\). Define \(\alpha = 1/(\gamma+\sigma)\). As \(\gamma + \sigma \to \infty\), \(\alpha \to 0\), so \((\theta_i/s_i)^\alpha \to 1\) for all \(i\). Hence \(p^A \to (1/\sqrt{N}, \ldots, 1/\sqrt{N})\)---the uniform direction.

A uniform price vector satisfies \(p_A't = 1/\sqrt{N} \cdot \sum_i t_i > 0\) for any \(t \in \mathbb S^{N-1}_+\) with at least one positive component. Hence the cone converges to the entire positive orthant (a hemisphere in \(\mathbb S^{N-1}\)).

Conversely, as \(\gamma + \sigma \to 0\), \(\alpha \to \infty\), and \(p^A\) concentrates on the task \(i^* = \arg\max_i \theta_i/s_i\). Only technologies with large \(t_{i^*}\) are adopted, making the cone very narrow.

\textbf{(2) Partial-adoption region narrows with \(\gamma + \sigma\).}

By Lemma \ref{lem:threshold-ordering}, \(\chi_{100}/\chi_0 > 1\) when \(t \not\propto x_A\). The ratio measures how much the supporting hyperplane rotates from \(x_A\) to \(t\): \[
\frac{\chi_{100}}{\chi_0} = \frac{\varrho(p_{100})/(p_{100}'t)}{\varrho(p_A)/(p_A't)}.
\] As \(\gamma + \sigma \to \infty\), both \(p_A\) and \(p_{100}\) converge to uniformity (the exponents \(1/(\gamma+\sigma)\) and \(1/\sigma\) in their definitions both shrink). When \(p_A \approx p_{100}\), the thresholds \(\chi_0\) and \(\chi_{100}\) nearly coincide, so \(\chi_{100}/\chi_0 \to 1\).

In the rigid limit (\(\gamma + \sigma\) small), \(p_A\) and \(p_{100}\) point in very different directions (one concentrated on the worker's bottleneck, the other on the technology's strength), widening the partial-adoption region.

\textbf{Individual parameter effects.} Holding other parameters fixed: (i) \(\gamma \to 0\) (rigid specialist) allows a wide partial-adoption region; \(\gamma \to \infty\) compresses it. (ii) \(\sigma \to 0\) (Leontief production) allows a wide region; \(\sigma \to \infty\) compresses it. Both effects work through the concentration of the price vectors \(p_A\) and \(p_{100}\). See Derivation M58 for the connection to the isoquant-scaling proof. \(\square\)

\end{document}